\documentclass[11pt,reqno]{amsart}

\usepackage[T1]{fontenc}
\usepackage[utf8]{inputenc}
\usepackage{times}
\usepackage{amsmath,amssymb,amsfonts,amsthm,mathtools}
\usepackage{bm}
\usepackage{booktabs}
\usepackage[margin=1.1in,includefoot,headheight=0pt,headsep=0pt,footskip=18pt]{geometry}
\usepackage[colorlinks=true,citecolor=blue,linkcolor=blue,urlcolor=blue]{hyperref}
\hypersetup{
  pdftitle={The Exact Second Generalized Covering Radius of Binary Primitive Triple-Error-Correcting BCH Codes},
  pdfauthor={Isaac Barouch Essayag; Aryeh Lev Zabokritskiy (Yohananov)},
  pdflang={en}
}

\date{}
\makeatletter
\def\ps@plain{\ps@empty
  \def\@oddfoot{\normalfont\small\hfil\thepage\hfil}
  \let\@evenfoot\@oddfoot}
\let\ps@firstpage\ps@plain
\makeatother
\AtBeginDocument{\pagestyle{plain}}
\allowdisplaybreaks
\newtheorem{theorem}{Theorem}[section]
\newtheorem{lemma}[theorem]{Lemma}
\newtheorem{proposition}[theorem]{Proposition}

\theoremstyle{definition}
\newtheorem{definition}[theorem]{Definition}

\theoremstyle{remark}
\newtheorem{remark}[theorem]{Remark}

\renewcommand{\leq}{\leqslant}

\renewcommand{\geq}{\geqslant}

\newcommand{\F}{\mathbb F}
\newcommand{\eqdef}{\triangleq}
\newcommand{\T}{\intercal}
\newcommand{\bch}{\mathrm{BCH}}
\newcommand{\bfa}{\bm{a}}
\newcommand{\bfb}{\bm{b}}
\newcommand{\bfs}{\bm{s}}
\newcommand{\bfv}{\bm{v}}
\newcommand{\bzero}{\bm{0}}
\DeclareMathOperator{\tr}{Tr}
\DeclareMathOperator{\wt}{wt}
\DeclareMathOperator{\supp}{supp}

\title{The Exact Second Generalized Covering Radius of Binary Primitive
Triple-Error-Correcting BCH Codes}

\author{Isaac Barouch Essayag}
\address{MIGAL -- Galilee Research Institute/Tel-Hai University of Kiryat
Shmona in the Galilee, P.O. Box 831, Kiryat Shmona 1101602, Israel}
\email{isaac.es@migal.org.il}

\author{Aryeh Lev Zabokritskiy (Yohananov)}
\address{Department of Computer Science, Tel-Hai University of Kiryat
Shmona in the Galilee, Israel}
\address{MIGAL -- Galilee Research Institute, P.O. Box 831,
Kiryat Shmona 1101602, Israel}
\email{yuhanalev@telhai.ac.il}
\subjclass[2020]{94B15, 11T71, 11G20}
\keywords{generalized covering radius, BCH codes, algebraic curves over
finite fields, deep cosets, computer-assisted proof}

\begin{document}

\begin{abstract}
Let \(C_m\eqdef\bch(3,m)\) be the binary primitive
triple-error-correcting BCH code of length \(2^m-1\).  We determine its
second generalized covering radius exactly: \(R_2(C_m)=8\) for every
\(m\geq5\).
Equivalently, every two-dimensional syndrome subspace is contained in
the binary span of at most eight parity-check columns, and eight columns
are necessary in the worst case.  This matches the known lower bound.
\end{abstract}

\maketitle

\section{Introduction}

The covering radius of a linear code measures the largest distance
from an ambient word to the code.  Its generalized version replaces
one ambient word by several words that must be covered simultaneously;
see~\cite{CohenHonkalaLitsynLobstein1997} for the classical background
and~\cite{ElimelechFirerSchwartz2021,ElimelechWeiSchwartz2022} for
generalized covering radii.

\begin{definition}
Let \(C\subseteq\F_2^n\) be a binary \([n,k]\) code with full-rank
parity-check matrix
\[
        H=(h_1,\ldots,h_n)\in\F_2^{(n-k)\times n}.
\]
For \(\bfv=(v_1,\ldots,v_n)\in\F_2^n\), its syndrome is
\(H\bfv^\T\).  Its support and Hamming weight are
\[
 \supp(\bfv)\eqdef\{i:v_i\neq0\},
 \qquad
 \wt(\bfv)\eqdef|\supp(\bfv)|.
\]
For \(t\geq1\), its \(t\)-th generalized covering radius is
\[
 R_t(C)\eqdef
 \min\left\{
 r:
 \begin{array}{l}
 \text{for every choice of }t\text{ syndromes }
 \bfs_1^\T,\ldots,\bfs_t^\T\in\F_2^{n-k}\\[-1mm]
 \text{there is }I\subseteq\{1,\ldots,n\}\text{ with }|I|\leq r\\[-1mm]
 \text{such that }\{\bfs_1^\T,\ldots,\bfs_t^\T\}\subseteq
 \operatorname{span}_{\F_2}\{h_i:i\in I\}
 \end{array}
 \right\}.
\]
\end{definition}

Thus \(R_1(C)\) is the ordinary covering radius, while \(R_2(C)\)
measures how efficiently two arbitrary syndromes can share
parity-check columns.  We say that a set
\(E\subseteq\{1,\ldots,n\}\) represents a syndrome \(\bfs^\T\) if
\[
        \bfs^\T=\sum_{i\in E}h_i.
\]
The least possible value of \(|E|\) is the \emph{coset weight} of the
syndrome, and a binary vector whose support is a minimum such set is a
coset leader.

Binary primitive BCH codes form a natural test family for this
question.  Let \(\bch(e,m)\) denote the binary primitive
narrow-sense BCH code of length \(2^m-1\) and designed distance
\(2e+1\).  For the triple-error-correcting code, put
\[
        C_m\eqdef\bch(3,m).
\]
\(\F_{2^m}^*\eqdef\F_{2^m}\setminus\{0\}\) denotes the nonzero field
elements.  A convenient parity-check column indexed by
\(x\in\F_{2^m}^*\) is
\[
        h(x)\eqdef(x,x^3,x^5)^\T.
\]
We identify the column \(h(x)\) with its field index \(x\).
Throughout, \(\F_{2^m}^3\) is viewed as a vector space over
\(\F_2\).  Expanding each field entry in a fixed
\(\F_2\)-basis turns the columns \(h(x)\) into a binary
parity-check matrix.  All column spans and sums below are over
\(\F_2\).
For \(m\geq5\), the binary cyclotomic cosets containing \(1,3,5\) are
distinct and each has size \(m\).  The matrix therefore has rank
\(3m\), and its syndrome space is \(\F_{2^m}^3\).
The BCH bound gives minimum distance at least seven.  General bounds
for the ordinary covering radius of long primitive BCH codes were
developed by Tiet\"av\"ainen and Cohen
~\cite{Tietavainen1987,Cohen1997PrimitiveBCH}.
The ordinary covering radius of \(\bch(3,m)\) is five for every
\(m\geq5\).  Van der Horst and Berger proved the case \(4\mid m\),
Assmus and Mattson proved the odd case, and Helleseth completed the
remaining class \(m\equiv2\pmod4\)
~\cite{VanDerHorstBerger1976,AssmusMattson1976,Helleseth1978}.
Covering two syndromes separately
therefore gives
\[
        R_2\!\left(\bch(3,m)\right)\leq10.
\]
Equivalently, every syndrome is represented by a set of at most five
parity-check columns, so two independent choices use at most ten columns.
On the other hand, Yohananov and Schwartz proved
\[
        R_2\!\left(\bch(e,m)\right)\geq3e-1
\]
whenever
\[
        2^{\lceil m/2\rceil}\geq2e-1
\]
~\cite[Theorem~3]{YohananovSchwartz2025}.  In particular,
\[
        R_2\!\left(\bch(3,m)\right)\geq8
        \qquad(m\geq5).
\]
The full problem for the triple-error-correcting family was therefore
to lower the unrestricted upper bound from ten to eight.

Recent work of \"Ozbudak and \"Ozt\"urk treats the restricted problem
in which both syndromes have zero first BCH coordinate.  They show that
such pairs are spanned by at most nine columns for odd \(m\geq11\) and
even \(m\geq20\)~\cite{OzbudakOzturk2026Second}.  Their construction
chooses two sets of five parity-check columns, one summing to each
syndrome, with a common column; their union therefore has size at most
nine.  For \(m\geq8\), our reduction shows that a counterexample to an
eight-column bound could occur only when the two syndromes and their sum
each require five columns and all three have zero first coordinate.  We
eliminate this remaining obstruction by finding two five-column
representations whose intersection contains at least two columns, so
their union has size at most eight.  This is precisely the extra
intersection needed to pass from nine columns to eight.

Related work has obtained exact values and upper and lower bounds
differing by one for the third
generalized covering radius of the double-error-correcting BCH family
~\cite{OzbudakOzturk2026Third}.  For the same family, Xiong and Yip
obtained a new lower bound for the fourth generalized radius, simplified
known lower bounds for the second and third radii, and bounded higher
radii when the field is sufficiently large relative to the order
~\cite{XiongYip2026}.  Exact second generalized covering radii are also
known for other families of binary cyclic codes
~\cite{LuoZhouMesnagerSagarYan2026}.

We now determine the unrestricted parameter for every \(m\geq5\).

\begin{theorem}
\label{thm:main}
For every integer \(m\geq5\),
\[
        R_2\!\left(\bch(3,m)\right)=8.
\]
\end{theorem}

The lower bound in Theorem~\ref{thm:main} is the \(e=3\) instance of
\cite[Theorem~3]{YohananovSchwartz2025}.  Our contribution is the
matching upper bound.  The proof has two complementary parts: an
algebraic argument for every \(m\geq17\), and deterministic exact
verification of the finite range \(5\leq m\leq16\).

The proof extends the argument based on intersecting sets of columns used for the
double-error-correcting family~\cite{YohananovSchwartz2025}.  There,
each syndrome is a sum of at most three columns, and one common column
reduces the union of the two representing sets from six columns to
five.  Here each of the two syndromes may require five columns, so the
exact bound eight requires two common columns.  This simple intersection
count is the central idea of the proof.

Section~\ref{sec:preliminaries} proves the structural reduction, and
Section~\ref{sec:fl-stable-range} proves the upper bound for \(m\geq17\).
Section~\ref{sec:finite-range-summary} states the certified theorem for the
finite range, after which the main proof and outlook conclude the main text.

\section{Preliminaries and the structural reduction}
\label{sec:preliminaries}

The purpose of this section is to remove all syndrome pairs that can
already be represented using a union of at most eight parity-check
columns.  After the basic definitions, we show that for \(m\geq8\)
only one special configuration requires further work.  The cases
\(m=5,6,7\) are instead verified directly on the whole syndrome space.

Put
\[
        q\eqdef2^m,\qquad
        F\eqdef\F_q,\qquad
        F^*\eqdef F\setminus\{0\},\qquad
        n\eqdef q-1.
\]
We retain \(C_m\) and the columns \(h(x)\) from the introduction and
index the coordinates of \(C_m\) by \(F^*\).  After expanding their field
entries in a fixed \(\F_2\)-basis, the columns \(h(x)\), \(x\in F^*\),
form a binary parity-check matrix \(H\) for \(C_m\).  We identify a
binary vector \(\bfv=(v_x)_{x\in F^*}\) with its support
\(E=\supp(\bfv)\), namely, the set of field indices of its selected
columns.  These field indices are also called \emph{locators}.  Via the
fixed \(\F_2\)-linear identification
\(\F_2^{3m}\cong F^3\), one has
\[
        H\bfv^\T=\sum_{x\in E}h(x),
        \qquad
        \wt(\bfv)=|E|.
\]
For \(\bfs^\T\in F^3\), define
\[
 \begin{aligned}
 \ell(\bfs^\T)
 &\eqdef
 \min\left\{
 \wt(\bfv):
 \bfv\in\F_2^n,\ H\bfv^\T=\bfs^\T
 \right\}\\
 &=
 \min\left\{
 |E|:
 E\subseteq F^*,\
 \bfs^\T=\sum_{x\in E}h(x)
 \right\}.
 \end{aligned}
\]
We call \(\ell(\bfs^\T)\) the coset weight of the syndrome.  A
\(\bfs^\T\)-leader is the support of a coset leader attaining this
minimum.  Since the ordinary covering radius is
five, a syndrome is called \emph{deep} if \(\ell(\bfs^\T)=5\), and
\emph{shallow} if \(\ell(\bfs^\T)\leq4\).
When a bold syndrome symbol occurs as a subscript, its terminal
transpose is understood and suppressed.
To count the columns needed to represent a fixed pair of syndromes
simultaneously, define
\[
 d_2(\bfa^\T,\bfb^\T)\eqdef
 \min\left\{
     |E\cup G|:
     E,G\subseteq F^*,\
     \bfa^\T=\sum_{x\in E}h(x),\
     \bfb^\T=\sum_{x\in G}h(x)
 \right\}.
\]
Thus \(d_2(\bfa^\T,\bfb^\T)\) is the least number of parity-check
columns needed to represent these two particular syndromes
simultaneously; a column used in both representations is counted only
once.  The second generalized covering radius is the worst such cost:
\[
        R_2(C_m)=
        \max_{\bfa^\T,\bfb^\T\in F^3}
        d_2(\bfa^\T,\bfb^\T).
\]

The locator normalization used throughout the proof is an exact
symmetry of this covering problem.

\begin{lemma}[Locator-scaling covariance]
\label{lem:locator-scaling}
For \(r\in F^*\), define
\[
 \nu_r\bigl((\sigma_1,\sigma_3,\sigma_5)^\T\bigr)
 \eqdef
 \left(\frac{\sigma_1}{r},
       \frac{\sigma_3}{r^3},
       \frac{\sigma_5}{r^5}\right)^\T
\]
and \(r^{-1}E\eqdef\{x/r:x\in E\}\) for \(E\subseteq F^*\).
Then
\[
 \sum_{x\in r^{-1}E}h(x)
 =\nu_r\left(\sum_{x\in E}h(x)\right).
\]
The map \(E\mapsto r^{-1}E\) is a bijection on supports, with inverse
\(E\mapsto rE\), and it preserves cardinalities, intersections,
unions, and symmetric differences.  Consequently,
\[
 \ell(\nu_r(\bfs^\T))=\ell(\bfs^\T),
 \qquad
 d_2(\nu_r(\bfa^\T),\nu_r(\bfb^\T))
 =d_2(\bfa^\T,\bfb^\T).
\]
In particular, scaling preserves deepness, coset-weight patterns, and
\(d_2\), and every normalized solution lifts back by
multiplying all locators by \(r\).
\end{lemma}

\begin{proof}
The displayed syndrome identity follows coordinatewise from
\(h(x/r)=(x/r,x^3/r^3,x^5/r^5)^\T\).  All remaining assertions follow
from the stated inverse bijection on \(F^*\).
\end{proof}

For linearly independent syndromes \(\bfa^\T,\bfb^\T\), let
\[
        W\eqdef
        \{\bzero^\T,\bfa^\T,\bfb^\T,(\bfa+\bfb)^\T\},
\]
be the two-dimensional binary syndrome space that they span.  Its
\emph{coset-weight pattern} is the nondecreasing rearrangement of
\[
        \bigl(
        \ell(\bfa^\T),
        \ell(\bfb^\T),
        \ell((\bfa+\bfb)^\T)
        \bigr).
\]
We abbreviate the patterns \((4,5,5)\) and \((5,5,5)\) by \(455\) and
\(555\), respectively.  A plane \(W\) has zero first coordinate if
\[
        W\subseteq\{0\}\times F^2.
\]
A \(555\) plane with zero first coordinate is such a space with
coset-weight pattern \(555\).

The absolute trace used below is
\[
 \tr\eqdef\operatorname{Tr}_{F/\F_2},
 \qquad
 \tr(z)=\sum_{i=0}^{m-1}z^{2^i}.
\]

With the basic objects now defined, we can describe the proof before
entering the reductions.  Adding two binary column representations
shows why the third nonzero syndrome \((\bfa+\bfb)^\T\) enters the
argument.  Choosing leaders for the two smallest coset weights among
the three nonzero syndromes costs at most their sum.  This either gives
an eight-column union immediately or leaves only the patterns \(455\)
and \(555\).  The first reduction below,
using the deep-syndrome characterization recalled from Charpin,
Helleseth, and Zinoviev, eliminates the \(455\) pattern and every
\(555\) space with a nonzero first coordinate.

Second, for a remaining \(555\) space with zero first coordinate, it is
enough to find two five-element leaders whose intersection has size at
least two.  Lemma~\ref{lem:locator-scaling} shows that dividing every
column index by \(r\in F^*\) sends a syndrome
\[
 (\sigma_1,\sigma_3,\sigma_5)^\T
 \quad\text{to}\quad
 \left(\frac{\sigma_1}{r},
       \frac{\sigma_3}{r^3},
       \frac{\sigma_5}{r^5}\right)^\T,
\]
without changing the covering problem; in particular, it preserves a
zero first coordinate and supplies the inverse lift.  If two
proposed common indices are distinct elements \(r,t\in F^*\), division
by \(r\) sends them to \(\{1,t/r\}\).  Writing \(t/r=1+s\) gives the
normalized pair \(\{1,1+s\}\), where \(s\notin\{0,1\}\).  For a normalized
syndrome \(\bfs^\T=(0,\alpha,\beta)^\T\), the two fixed indices leave
three indices \(\gamma_1,\gamma_2,\gamma_3\) to be found, subject to
\[
 h(1)+h(1+s)+\sum_{i=1}^3h(\gamma_i)=\bfs^\T.
\]
After translating \(\gamma_i=\zeta_i+s\), Newton's identities convert
the three power-sum equations in this display into the completing cubic
of Lemma~\ref{lem:fl-completing-cubic}.  Translating its roots back by
\(\gamma_i=\zeta_i+s\) gives the three remaining column indices.  Thus
a single value of \(s\) for which both cubics split produces the two
desired leaders and hence an eight-column union.

Finally, for \(m\geq17\) we construct an algebraic curve whose
\(F\)-rational points encode values of \(s\) for which both cubics
split.  A point on this curve supplies both three-column completions
for the same \(s\).  Hasse--Weil guarantees such a point for
\(m\geq18\), and a separate trace argument treats \(m=17\).  For the
finite range, deterministic
verification checks the same shared-column conclusion directly, after
the reductions and available field symmetries.

\subsection{Reduction to 555 planes with zero first coordinate}

For sets \(E,G\subseteq F^*\), define their symmetric difference by
\[
 E\mathbin{\triangle}G
 \eqdef
 (E\setminus G)\cup(G\setminus E),
\]
the set of indices belonging to exactly one of \(E\) and \(G\).  If
\(E\) and \(G\) represent \(\bfa^\T\) and \(\bfb^\T\), respectively,
then the three supports
\[
        E,\qquad G,\qquad E\mathbin{\triangle}G
\]
represent \(\bfa^\T,\bfb^\T,(\bfa+\bfb)^\T\), respectively.  When
\(\bfa^\T\) and \(\bfb^\T\) are linearly independent, these are the
three nonzero syndromes in
\[
        \{\bzero^\T,\bfa^\T,\bfb^\T,(\bfa+\bfb)^\T\}.
\]
Moreover,
\[
        2|E\cup G|
        =
        |E|+|G|+|E\mathbin{\triangle}G|.
\]
Changing the ordered basis of \(W\) by
\(\operatorname{GL}_2(\F_2)\) permutes
\(\bfa^\T,\bfb^\T,(\bfa+\bfb)^\T\) and preserves the support union.

\begin{lemma}[The only possible weight patterns]
\label{lem:only-455-555}
Let \(m\geq5\).  If
\[
        d_2(\bfa^\T,\bfb^\T)>8,
\]
then \(\bfa^\T\) and \(\bfb^\T\) are linearly independent, and their
syndrome plane has coset-weight pattern \(455\) or \(555\).
\end{lemma}

\begin{proof}
A dependent pair is covered by one coset leader of weight at most five.
For an independent pair, let the three nonzero coset weights be
\(a\leq b\leq c\).  Choosing leaders for the two syndromes of weights
\(a\) and \(b\), and then changing the ordered basis of their plane if
necessary, gives
\[
        d_2(\bfa^\T,\bfb^\T)\leq a+b.
\]
Since the ordinary covering radius is five, the inequality
\(d_2>8\) leaves only \((a,b,c)=(4,5,5)\) or \((5,5,5)\).
\end{proof}

We use the following published classification, rewritten in our
syndrome notation.

\begin{lemma}[Charpin--Helleseth--Zinoviev deep-syndrome classification]
\label{lem:ca-chz-classification}
Let \(m\geq8\), and write a syndrome with nonzero first coordinate as
\[
        \bfs^\T=(\sigma_1,\sigma_3,\sigma_5)^\T,
        \qquad \sigma_1\neq0.
\]
Put
\[
        \rho\eqdef\sigma_1,
        \qquad
        \eta\eqdef\frac{\sigma_3}{\rho^3}+1,
        \qquad
        M(\bfs^\T)
        \eqdef
        \sigma_3^2
        +\sigma_1^3\sigma_3
        +\sigma_1\sigma_5
        +\sigma_1^6.
\]
Then \(\bfs^\T\) is deep if and only if
\[
        \tr(\eta)=1
        \qquad\text{and}\qquad
        M(\bfs^\T)=0.
\]
Equivalently, the deep syndromes with nonzero first coordinate are precisely
\[
 \Phi(\rho,\eta)\eqdef
 \bigl(\rho,\ \rho^3(1+\eta),
 \rho^5(1+\eta+\eta^2)\bigr)^\T,
 \qquad
 \rho\in F^*,\quad \tr(\eta)=1.
\]
In particular, \(M(\bfs^\T)\neq0\) implies
\(\ell(\bfs^\T)\leq4\).
\end{lemma}

For \(m\geq10\), the equivalence in
Lemma~\ref{lem:ca-chz-classification} is precisely
Theorems~2 and~3 of Charpin, Helleseth, and
Zinoviev~\cite[Theorems~2 and~3]{CharpinHellesethZinoviev2006}.
For \(m=8,9\), it follows by combining
\cite[Theorem~1, Theorem~2, Remark~1, and
Lemma~1]{CharpinHellesethZinoviev2006}.  The family with nonzero first coordinate in
their Theorem~2 and Remark~1 has cardinality \(n(n+1)/2\), while their
Lemma~1 counts \((5n^2+13n)/6\) deep cosets with zero first coordinate, and
\[
        \frac{n(n+1)}2+\frac{5n^2+13n}{6}
        =
        \frac{4n(n+2)}3.
\]
The right-hand side is the full number of deep cosets in their
Theorem~1.  Thus no additional deep syndromes with nonzero first coordinate remain.

The parameters are unique: \(\rho\) is the first syndrome coordinate,
while \(\eta\) is determined by the coordinate \(\sigma_3\).

We need the classification in order to force a prescribed column into
a five-element leader.  Let the deep syndrome
\(\bfs^\T\) have nonzero first coordinate \(\rho\), and choose
\(\beta\in F^*\setminus\{\rho\}\).  The residual syndrome
\(\bfs^\T+h(\beta)\) then has nonzero first coordinate.  If its value
of \(M\) is nonzero, Lemma~\ref{lem:ca-chz-classification} shows that
it has a leader using at most four columns; adjoining \(h(\beta)\)
represents \(\bfs^\T\) with at most five columns.  Thus it is enough to
prove that \(M(\bfs^\T+h(\beta))\neq0\).

We shall also use the following published fact twice.  A
nonzero syndrome whose first coordinate is zero cannot have coset
weight one, two, or four; consequently, if its coset weight is at most
four, then it has coset weight three.  The exclusion of coset weight
four is proved in the proof of Lemma~1 of Charpin, Helleseth, and
Zinoviev \cite[Lemma~1]{CharpinHellesethZinoviev2006}, where the result
is attributed to van der Horst and Berger.  The exclusions of weights
one and two are immediate from the first syndrome coordinate.

\begin{lemma}[Reduction to \(555\) planes with zero first coordinate]
\label{lem:ca-zero-first-555}
Let \(m\geq8\).  If
\[
        d_2(\bfa^\T,\bfb^\T)>8,
\]
then
\[
        \ell(\bfa^\T)
        =\ell(\bfb^\T)
        =\ell((\bfa+\bfb)^\T)
        =5
\]
and
\[
        (\bfa^\T)_1
        =(\bfb^\T)_1
        =((\bfa+\bfb)^\T)_1
        =0.
\]
In other words, any hypothetical obstruction surviving the
noncomputational reductions would be a two-dimensional syndrome space
whose three nonzero elements are deep and have first coordinate zero.
\end{lemma}

\begin{proof}
Lemma~\ref{lem:only-455-555} leaves only the patterns \(455\) and
\(555\).

We first show that, for a deep syndrome with nonzero first coordinate, only the column
indexed by its first coordinate can leave a deep residual.  If
\(\bfs^\T=\Phi(\rho,\eta)\) is deep and \(\beta\neq\rho\), put
\[
        \vartheta\eqdef1+\frac{\beta}{\rho}.
\]
A direct substitution in the polynomial \(M\) gives
\[
 \frac{M(\bfs^\T+h(\beta))}{\rho^6\eta}
 =
 (\vartheta+1)(\vartheta^2+\vartheta+\eta).
\]
The value \(\vartheta=1\) would give the excluded coordinate \(\beta=0\), while
\(\vartheta^2+\vartheta+\eta\) cannot vanish because
\(\tr(\eta)=1\).  Thus \(M(\bfs^\T+h(\beta))\neq0\), and
Lemma~\ref{lem:ca-chz-classification} shows that
\(\bfs^\T+h(\beta)\) is not deep, so
\(\ell(\bfs^\T+h(\beta))\leq4\).  On the other hand, if
\(\ell(\bfs^\T+h(\beta))\leq3\), adjoining the locator \(\beta\) would
represent the deep syndrome \(\bfs^\T\) with at most four columns.  Hence
\[
        \ell(\bfs^\T+h(\beta))=4
        \qquad(\beta\in F^*\setminus\{\rho\}),
\]
whereas, after scaling, the residual at \(\beta=\rho\) is
\[
        (0,\eta,\eta+\eta^2)^\T.
\]
By Newton's identities, a three-element representation would have
locator polynomial
\[
        x^3+(1+\eta)x+\eta
        =
        (x+1)(x^2+x+\eta),
\]
whose quadratic factor is irreducible because
\(\tr(\eta)=1\).  Since \(\eta\neq0\), the residual is nonzero.
The elementary fact above for zero first coordinate, together with covering radius
five, shows that the residual is deep.  Hence \(\rho\) is the only
column index whose removal can leave a deep residual.

Consider a \((4,5,5)\) pattern.  Write
\[
        \ell((\bfa+\bfb)^\T)=4,\qquad
        \ell(\bfa^\T)=\ell(\bfb^\T)=5.
\]
By the same elementary fact, the syndrome of coset weight
four has nonzero first coordinate.  Hence at least one of
\(\bfa^\T,\bfb^\T\), say \(\bfa^\T\), has nonzero first coordinate
\(\rho\).  Choose a four-element leader \(T\) for
\((\bfa+\bfb)^\T\), and choose
\[
        \beta\in T\setminus\{\rho\}.
\]
The conclusion just proved gives a four-element leader
\(E_{\bfa}\) for \(\bfa^\T+h(\beta)\).  Deepness of \(\bfa^\T\) forces
\(\beta\notin E_{\bfa}\).  Therefore
\[
        E\eqdef E_{\bfa}\cup\{\beta\}
\]
is a five-element leader for \(\bfa^\T\), and it meets \(T\).  The support
\[
        G\eqdef E\mathbin{\triangle}T
\]
represents \(\bfb^\T\), and
\[
        |E\cup G|
        =
        |E\cup T|
        \leq5+4-1=8.
\]
Thus the \((4,5,5)\) pattern is never an obstruction.

It remains to consider a \((5,5,5)\) plane.  The first-coordinate map
on its two-dimensional syndrome space has rank two, one, or zero;
accordingly, none, exactly one, or all three of the nonzero syndromes
have first coordinate zero.  Suppose first that
all three first coordinates are nonzero.  After scaling, write
\[
        \bfa^\T=\Phi(1,\eta),\qquad
        \bfb^\T=\Phi(\tau,\theta),
\]
where
\[
        \tau\notin\{0,1\},
        \qquad
        \tr(\eta)
        =
        \tr(\theta)
        =1.
\]
With
\[
        \xi\eqdef\frac{\eta+\tau^2\theta}{\tau(1+\tau)},
\]
direct expansion gives
\[
        M((\bfa+\bfb)^\T)
        =
        \tau^3(1+\tau)^2(\xi^2+\xi+\theta).
\]
Deepness of \((\bfa+\bfb)^\T\) would force
\(\theta=\xi^2+\xi\), contradicting
\(\tr(\theta)=1\).

Finally suppose that exactly one first coordinate is zero.  After
permuting the three nonzero syndromes, the other two have the form
\[
        \Phi(\rho,\eta),\qquad \Phi(\rho,\theta).
\]
Their sum is
\[
        \bigl(0,\rho^3\varepsilon,
        \rho^5(\varepsilon+\varepsilon^2)\bigr)^\T,
        \qquad \varepsilon\eqdef \eta+\theta,
        \qquad \tr(\varepsilon)=0.
\]
If \(\varepsilon=0\), the sum is zero, contrary to independence of the
two generators.  Otherwise choose \(\zeta\in F\) with
\[
        \zeta^2+\zeta=\varepsilon.
\]
The three distinct nonzero locators
\[
        \rho,\qquad \rho\zeta,\qquad \rho(\zeta+1)
\]
represent this syndrome with zero first coordinate, contradicting its deepness.  Hence
every \(555\) plane has all three first coordinates zero.
\end{proof}

\begin{remark}
It remains to show that every \(555\) plane with zero first coordinate
is contained in the span of at most eight parity-check columns.
\end{remark}

\section{An algebraic proof for
\texorpdfstring{\(m\geq17\)}{m >= 17}}
\label{sec:fl-stable-range}

We retain all notation from Section~\ref{sec:preliminaries}.
For later formulas we extend this notation by
\[
        h(0)\eqdef\bzero^\T.
\]
The computer-assisted part of the argument proves the exact finite
range \(5\leq m\leq16\); see
Theorem~\ref{thm:ca-certified-finite-range}.  The purpose of the
present section is to give an algebraic proof for \(m\geq17\).
The algebraic argument and the finite checks for \(m\geq8\) use the
structural reduction in Lemma~\ref{lem:ca-zero-first-555}; its hypothesis
is automatic in the present section.  According to that lemma, the only possible
counterexample to the eight-column upper bound is a two-dimensional
syndrome space whose
three nonzero elements are deep and have first coordinate zero.

We also use the published lower bound
\begin{equation}
\label{eq:fl-published-lower-bound}
        R_2(C_m)\geq8,
\end{equation}
which is the \(t=2\), triple-error-correcting instance of
\cite[Theorem~3]{YohananovSchwartz2025}.  That theorem assumes
\[
        2^{\lceil m/2\rceil}\geq 2e-1.
\]
For \(e=3\) this holds for every \(m\geq5\).  The proof below supplies
the matching upper bound for \(m\geq17\).

\subsection{The fixed pair and its completing cubic}

The goal is to find, in every \(555\) plane with zero first coordinate, two
weight-five leaders that share a pair of locators.  We normalize the
pair and express each three-locator completion by one cubic.  We then
seek one value of the common parameter for which both cubics split.
The proof separates cases according to whether the third power sums
of the two syndromes coincide, because the corresponding
point-counting bounds are different.

Let
\[
        W\eqdef
        \{\bzero^\T,\bfa^\T,\bfb^\T,(\bfa+\bfb)^\T\}
\]
be a two-dimensional syndrome space such that its three nonzero
elements are deep and have first coordinate zero.  We write
\[
        \bfa^\T\eqdef(0,a,b)^\T,
        \qquad
        \bfb^\T\eqdef(0,c,d)^\T.
\]
The map
\[
 \pi_3\colon\{0\}\times F^2\longrightarrow F,
 \qquad
 \pi_3\bigl((0,u,v)^\T\bigr)\eqdef u
\]
is the projection onto the third power sum.  Its rank
on \(W\), always taken over \(\F_2\), records how many distinct
third power sums occur.  This is the case distinction used in
the simultaneous point count below.

Choose a nonzero locator \(r\).  By
Lemma~\ref{lem:locator-scaling}, dividing all locators by \(r\) gives
\[
 \nu_r\bigl((0,u,v)^\T\bigr)
 =
 \left(0,\frac{u}{r^3},\frac{v}{r^5}\right)^\T.
\]
In particular, put
\[
 \begin{aligned}
 \bfa_r^\T&\eqdef\nu_r(\bfa^\T)
      =\left(0,a_r,b_r\right)^\T
      =\left(0,\frac a{r^3},\frac b{r^5}\right)^\T,\\
 \bfb_r^\T&\eqdef\nu_r(\bfb^\T)
      =\left(0,c_r,d_r\right)^\T
      =\left(0,\frac c{r^3},\frac d{r^5}\right)^\T.
 \end{aligned}
\]
All conclusions obtained after this normalization lift to the original
syndromes by multiplying every locator by \(r\).
After this normalization we prescribe the common pair
\[
        \{1,1+s\}.
\]
It consists of two distinct nonzero locators precisely when
\[
        s\notin\{0,1\}.
\]
The next lemma uses the standard locator-polynomial and Newton-identity
method from complete BCH decoding; see, for example,
\cite[Appendix~D]{VanDerHorstBerger1976}.  The normalization to a
prescribed common pair is specific to the present simultaneous problem.

\begin{lemma}[From a common pair to a cubic]
\label{lem:fl-completing-cubic}
Let \(\bfs^\T\eqdef(0,\alpha,\beta)^\T\), and put
\[
        D_{\bfs}(s)\eqdef\alpha+s^2+s,\qquad
        N_{\bfs}(s)\eqdef\beta+s^4+s.
\]
Assume that \(D_{\bfs}(s)\neq0\).  For
\(\gamma_1,\gamma_2,\gamma_3\) in an algebraic closure of \(F\), the
column-sum equation
\[
 h(1)+h(1+s)+\sum_{i=1}^3h(\gamma_i)=\bfs^\T
\]
holds if and only if the translated elements
\(\zeta_i\eqdef\gamma_i+s\) are the three roots, counted with
multiplicity, of
\[
        f_{\bfs}(y)
        \eqdef
        y^3+\frac{N_{\bfs}(s)}{D_{\bfs}(s)}y+D_{\bfs}(s).
\]
Consequently, the prescribed pair has a completing triple in \(F\) if
and only if \(f_{\bfs}\) splits completely over \(F\).
\end{lemma}

\begin{proof}
Let \(\gamma_1,\gamma_2,\gamma_3\) be the completing locators, put
\[
        p_j\eqdef\sum_{i=1}^3\gamma_i^j,
\]
and let \(e_1,e_2,e_3\) be their elementary symmetric functions.  The
value \(e_1=s\) is forced by the zero first coordinate, because
the prescribed pair also has sum \(s\).  Newton's identities in
characteristic two give
\[
        p_3=s^3+se_2+e_3
\]
and
\[
        p_5=s^5+e_2s^3+e_2^2s+e_2e_3+e_3s^2
\]
for the third and fifth power sums of the completing triple.  Adding
the contributions of \(1\) and \(1+s\), and then comparing with
\((0,\alpha,\beta)^\T\), gives
\[
        se_2+e_3=D_{\bfs}(s)
\]
and
\[
        N_{\bfs}(s)=D_{\bfs}(s)(e_2+s^2).
\]
When \(D_{\bfs}(s)\neq0\), translation of the triple polynomial
\[
        x^3+sx^2+e_2x+e_3
\]
by \(x=y+s\) therefore yields
\[
        y^3+(e_2+s^2)y+(se_2+e_3)
        =
        y^3+\frac{N_{\bfs}(s)}{D_{\bfs}(s)}y+D_{\bfs}(s).
\]
This proves the asserted equivalence and the translation rule for the
roots.
\end{proof}

Such a completing triple gives a valid five-column representation only
when its elements are distinct and nonzero and avoid the prescribed
pair.

Lemma~\ref{lem:fl-completing-cubic} is now applied twice, to the two
normalized syndromes \(\bfa_r^\T\) and \(\bfb_r^\T\).  For fixed \(r\)
and \(s\), it produces the two cubics
\[
        f_{\bfa_r}(y)
        \qquad\text{and}\qquad
        f_{\bfb_r}(y).
\]
If both cubics split for the same value of \(s\) and both completing
triples satisfy this validity condition, translating their roots back
as in Lemma~\ref{lem:fl-completing-cubic} gives two five-column
representations sharing the prescribed pair
\(\{1,1+s\}\).  The next steps first describe the splitting condition
for one cubic and then impose the two conditions simultaneously.

Before constructing the algebraic curve, we must exclude a degenerate
normalization in which every split cubic repeats one of the two fixed
columns.  For \(\bfs^\T\eqdef(0,\alpha,\beta)^\T\), define the quantity
\[
        \kappa_{\bfs}\eqdef\beta+\alpha^2+\alpha.
\]
The identity
\begin{equation}
\label{eq:fl-N-D-kappa}
        N_{\bfs}(s)
        =D_{\bfs}(s)^2+D_{\bfs}(s)+\kappa_{\bfs}
\end{equation}
shows that \(D_{\bfs}(s)\) and \(N_{\bfs}(s)\) are coprime exactly when
\(\kappa_{\bfs}\neq0\).  The next lemma shows that
\(\kappa_{\bfs}=0\) never yields a valid five-column representation in
this construction.

\begin{lemma}[Exclusion of a repeated fixed locator]
\label{lem:fl-kappa-zero}
Let \(\bfs_0^\T\eqdef(0,a,b)^\T\neq\mathbf0\), let \(r\in F^*\), and write
\[
        \bfs_r^\T\eqdef\nu_r(\bfs_0^\T)
        =(0,\alpha,\beta)^\T.
\]
Suppose \(\kappa_{\bfs_r}=0\).  For every \(s\) with
\(D_{\bfs_r}(s)\neq0\), a split
completing cubic repeats the locator \(1+s\) from the prescribed pair
and hence does not produce a five-element support.

For the original syndrome \(\bfs_0^\T\), this condition is
\begin{equation}
\label{eq:fl-bad-r-kappa}
        ar^3+br+a^2=0.
\end{equation}
It excludes at most three nonzero values of \(r\).
\end{lemma}

\begin{proof}
If \(\kappa_{\bfs_r}=0\), then
\[
        N_{\bfs_r}(s)
        =D_{\bfs_r}(s)^2+D_{\bfs_r}(s),
\]
and Lemma~\ref{lem:fl-completing-cubic} gives
\[
        f_{\bfs_r}(y)
        =
        y^3+\bigl(D_{\bfs_r}(s)+1\bigr)y+D_{\bfs_r}(s)
        =
        (y+1)\bigl(y^2+y+D_{\bfs_r}(s)\bigr).
\]
The root \(1\) translates back to \(1+s\), which is already a
member of the fixed pair.  If the quadratic factor splits, every
split specialization therefore repeats this fixed locator; if it does not
split, there is no completing triple over \(F\).  The values with
\(D_{\bfs_r}(s)=0\) lie outside the domain of
Lemma~\ref{lem:fl-completing-cubic}.

After replacing \((a,b)\) by \((a/r^3,b/r^5)\), multiplying
\(\kappa_{\bfs_r}=0\) by \(r^6\) gives
\eqref{eq:fl-bad-r-kappa}.  If \(a\neq0\), this is a nonzero cubic.
If \(a=0\), then \(b\neq0\), because \(\bfs_0^\T\) is nonzero, and the equation
reduces to \(br=0\).  In either case there are at most three nonzero
solutions.
\end{proof}

For two syndromes, Lemma~\ref{lem:fl-kappa-zero} excludes at most six
values of \(r\).  We henceforth choose \(r\) outside this exceptional
set.  For one normalized syndrome we abbreviate
\[
        \bfs_r^\T\eqdef(0,\alpha,\beta)^\T,\qquad
        D\eqdef\alpha+s^2+s,\qquad
        N\eqdef\beta+s^4+s,\qquad
        \kappa\eqdef\beta+\alpha^2+\alpha\neq0.
\]

The remaining algebraic argument has four steps.  First, we determine
the degree and genus of the splitting cover for one completing cubic.
Second, we combine the two one-syndrome covers and compute the genus of
the simultaneous cover.  Third, when the two third power sums
coincide, we choose \(r\) to avoid the exceptional normalizations.
Finally, Hasse--Weil supplies a rational point corresponding to a valid
shared pair.  Lemmas~\ref{lem:fl-quadratic-resolvent} and
\ref{lem:fl-cardano-kummer} in
Appendix~\ref{sec:fl-ramification-details} serve only the first step:
together they locate all ramification needed to compute the
one-syndrome genus.

\subsection{The splitting cover of one completing cubic}

Let \(\overline F\) be an algebraic closure of \(F\), and let \(s\) be
transcendental over \(\overline F\).  Put
\[
        \lambda\eqdef\frac ND,\qquad \mu\eqdef D,
\]
so that the generic completing cubic over \(\overline F(s)\) is
\(y^3+\lambda y+\mu\).  By its \emph{splitting cover} over \(F\) we mean the
smooth projective curve whose function field is the splitting field of
this cubic over \(F(s)\).  The cover is called \emph{regular} when \(F\)
is algebraically closed in that function field.

The technical ramification calculation is deferred to
Appendix~\ref{sec:fl-ramification-details}.  It uses the standard
function-field conventions and ramification formulas
from~\cite{Stichtenoth2009}.

\begin{proposition}[The one-syndrome splitting cover]
\label{prop:fl-one-cover}
Let \(\bfs_r^\T=(0,\alpha,\beta)^\T\) satisfy
\(\kappa_{\bfs_r}\neq0\), and let \(f_{\bfs_r}\in F(s)[y]\) be the completing
cubic from Lemma~\ref{lem:fl-completing-cubic}.  The cubic is
geometrically irreducible, its geometric and arithmetic Galois groups
are \(S_3\), and its smooth projective
splitting cover \(\mathcal X_{\bfs_r}\) is regular of degree six and genus
thirteen.
Its only nontrivial inertia groups are \(C_2\) above the two zeros of
\(D\); the corresponding reduced pole order is five and the local
different exponent is six.  The cover is unramified above infinity; in
the residue field of its quadratic-resolvent subextension, the variable
\(\eta\) satisfies
\[
        \eta^2+\eta=\alpha.
\]
\end{proposition}

The ramification lemmas and the proof of
Proposition~\ref{prop:fl-one-cover} are given in
Appendix~\ref{sec:fl-ramification-details}.

\subsection{Simultaneous splitting covers}

For the normalized syndromes \(\bfa_r^\T\) and \(\bfb_r^\T\) defined
above, let \(\mathcal X_{\bfa,\bfb,r}\) be the smooth projective cover obtained by composing
their two splitting fields.  The fiber above a base value \(s_0\) is
the set of points of the cover mapping to \(s_0\).  A rational point
away from the excluded fibers gives one value of \(s\) for which both
completing cubics split.
More precisely, call an \(F\)-rational point of
\(\mathcal X_{\bfa,\bfb,r}\) good if
its image \(s\) is finite, lies outside \(\{0,1\}\), and is outside the
branch locus, namely, the set of base values above which the cover
ramifies.  The decomposition group of a point is the subgroup of the
Galois group that fixes it.  Above an unramified \(F\)-rational base
value, an \(F\)-rational point on this regular Galois cover has trivial
decomposition group.  Both specialized cubics therefore split
completely over \(F\).
We now compute the genus according to the rank of \(\pi_3(W)\).

\subsubsection{Distinct third power sums}

\begin{proposition}[Distinct third power sums]
\label{prop:fl-distinct-a-genus}
Suppose \(a_r\neq c_r\) and
\(\kappa_{\bfa_r}\kappa_{\bfb_r}\neq0\).  Then
\(\mathcal X_{\bfa,\bfb,r}\) is a regular Galois cover of degree \(36\), with group
\[
        S_3\times S_3
\]
and genus
\[
        g(\mathcal X_{\bfa,\bfb,r})=181.
\]
\end{proposition}

\begin{proof}
By Proposition~\ref{prop:fl-one-cover}, the two factors are regular
\(S_3\)-covers, and their branch loci are
\[
        s^2+s+a_r=0
        \quad\text{and}\quad
        s^2+s+c_r=0.
\]
They are disjoint because their difference is the nonzero constant
\(a_r+c_r\).  A nontrivial intersection of the two Galois splitting
fields would give a nontrivial common quotient of two copies of
\(S_3\).  Such a quotient is either quadratic or \(S_3\), and its
ramification locus would have to be contained in both of the disjoint
branch loci.  It would therefore define a nontrivial everywhere
unramified geometric cover of \(\mathbb P^1\).  Over an algebraically
closed field, \(\mathbb P^1\) admits no nontrivial connected finite
\'{e}tale cover, so this is impossible.
The two splitting fields are consequently linearly disjoint, and the
composite has degree \(36\) and group \(S_3\times S_3\).  The same
geometric degree \(36\) is already the largest possible arithmetic
degree inside \(S_3\times S_3\).  The arithmetic group is therefore
the same product, and no constant-field extension occurs.

There are four geometric branch points.  At each, the inertia group
has order two and the local different exponent is six.  Thus there
are \(36/2=18\) points above each branch point, and
Riemann--Hurwitz gives
\[
\begin{aligned}
        2g(\mathcal X_{\bfa,\bfb,r})-2
        &=-2\cdot36+4\left(\frac{36}{2}\right)6\\
        &=-72+432=360.
\end{aligned}
\]
Therefore \(g(\mathcal X_{\bfa,\bfb,r})=181\).
\end{proof}

\subsubsection{Equal third power sums}

Suppose now that \(a_r=c_r\).  The two completing cubics then have the
common denominator
\[
        D\eqdef s^2+s+a_r.
\]
For each completing cubic \(y^3+\lambda y+\mu\), its quadratic-resolvent
subfield is defined by
\[
 z^2+z=1+\frac{\lambda^3}{\mu^2},
\]
as shown in Lemma~\ref{lem:fl-quadratic-resolvent}.
Adding a rational function of the form \(u^2+u\) to the right-hand side
amounts to replacing \(z\) by \(z+u\), and leaves the subfield unchanged.
In Appendix~\ref{sec:fl-ramification-details}, we use these
Artin--Schreier changes to choose representatives with minimal odd pole
orders; these are the reduced quadratic characters used below.  Put
\[
        K\eqdef\kappa_{\bfa_r},\qquad
        e_r\eqdef b_r+d_r,\qquad
        \kappa_{\bfb_r}=K+e_r.
\]
Since \(\bfa_r^\T\neq\bfb_r^\T\), one has \(e_r\neq0\).  At either common branch
point, the coefficient of \(D^{-5}\) in the sum of the two reduced
quadratic characters is
\[
        L
        \eqdef
        K^3+(K+e_r)^3
        =
        e_r(K^2+e_rK+e_r^2).
\]

\begin{proposition}[Equal third power sums]
\label{prop:fl-equal-a-genus}
Assume \(K(K+e_r)\neq0\).

\begin{enumerate}
\item If \(L\neq0\), then \(\mathcal X_{\bfa,\bfb,r}\) is regular of degree \(36\).
      At each of its two branch points, the three nontrivial
      quadratic characters have reduced pole orders \(5,5,5\), and
      \[
              g(\mathcal X_{\bfa,\bfb,r})=127.
      \]
\item Suppose \(a_r=0\), \(L=0\), and \(e_r\neq1\).  Then
      \(\mathcal X_{\bfa,\bfb,r}\) is regular of degree \(36\), the three
      reduced pole orders are \(5,5,3\), and
      \[
              g(\mathcal X_{\bfa,\bfb,r})=109.
      \]
\end{enumerate}
\end{proposition}

\begin{proof}
Proposition~\ref{prop:fl-one-cover} supplies two regular
\(S_3\)-covers with the same two branch points.  If \(L\neq0\), the
two individual quadratic characters and their sum
all have reduced pole order five.  The quadratic subfields are
distinct, and hence the local inertia in the degree-\(36\) composite
is the Klein four group \(V_4\cong C_2\times C_2\).  By the
different-exponent formula for totally ramified Artin--Schreier
composita \cite[Theorem~4.6 and Equation~(4.7)]{WuScheidler2010}, its
local different exponent is the sum of the different exponents of the
three quadratic subextensions:
\[
        (5+1)+(5+1)+(5+1)=18.
\]
There are \(36/4=9\) points above each of the two branch points.
Therefore
\[
\begin{aligned}
        2g(\mathcal X_{\bfa,\bfb,r})-2
        &=-72+2\left(\frac{36}{4}\right)18\\
        &=-72+324=252,
\end{aligned}
\]
which gives \(g(\mathcal X_{\bfa,\bfb,r})=127\).  The nontrivial sum character also
shows that the two quadratic subfields are different.  Any nontrivial
Galois intersection of two \(S_3\)-extensions is either their
quadratic subfield or the full \(S_3\)-extension; the latter would
also identify the quadratic subfields.  The splitting fields are
therefore linearly disjoint.  Their geometric degree is \(36\), the
largest possible arithmetic degree, so arithmetic and geometric
regularity follow.

It remains to examine the second case, where \(a_r=0\) and the leading
\(D^{-5}\) terms cancel at both common branch points.  Here
\[
        D=s^2+s,\qquad W_0\eqdef s^4+s=D^2+D.
\]
When \(L=0\), direct subtraction, which is addition in characteristic
two, gives the sum of the two quadratic-resolvent classes as
\[
        \frac{e_rW_0^2+e_r^2W_0}{D^5}
        =
        e_r^2D^{-4}
        +(e_r+e_r^2)D^{-3}
        +e_rD^{-1}.
\]
Adding the Artin--Schreier coboundaries of
\[
        \frac{e_r}{D^2}
        \qquad\text{and}\qquad
        \frac{\sqrt{e_r}}{D}
\]
leaves the reduced polar part
\[
        (e_r+e_r^2)D^{-3}
        +(e_r+\sqrt{e_r})D^{-1}.
\]
If \(e_r\neq1\), the coefficient of \(D^{-3}\) is nonzero, so the sum
character has exact reduced pole order three.  The local different
exponent is now
\[
        (5+1)+(5+1)+(3+1)=16.
\]
Again the sum character is nontrivial, so the local inertia is \(V_4\)
and the global splitting fields are linearly disjoint.  Hence
\[
\begin{aligned}
        2g(\mathcal X_{\bfa,\bfb,r})-2
        &=-72+2\left(\frac{36}{4}\right)16\\
        &=-72+288=216,
\end{aligned}
\]
and \(g(\mathcal X_{\bfa,\bfb,r})=109\).
\end{proof}

The exceptional equations in
Proposition~\ref{prop:fl-equal-a-genus} can be avoided uniformly.

\begin{lemma}[Choice of locator for projection rank at most one]
\label{lem:fl-low-rank-locator}
Let \(W\) be a \(555\) plane with zero first coordinate, and suppose
\(\dim_{\F_2}\pi_3(W)\leq1\).  If \(m\geq16\), there is
a choice of basis \(\bfa^\T,\bfb^\T\) of \(W\) and a nonzero locator \(r\) for
which the simultaneous cover is regular and has genus at most \(127\).
\end{lemma}

\begin{proof}
First suppose the projection has rank one.  Choose as \(\bfa^\T\) and
\(\bfb^\T\)
the two nonzero syndromes outside its kernel.  Then
\[
        \bfa^\T=(0,a,b)^\T,\qquad
        \bfb^\T=(0,a,d)^\T,
        \qquad a\neq0,\quad b+d\neq0.
\]
The conditions \(K=0\) and \(K+e_r=0\) are respectively
\[
        ar^3+br+a^2=0,
\qquad
        ar^3+dr+a^2=0,
\]
and together exclude at most six values of \(r\).  If \(L=0\), then
\[
        K/e_r\in\{\omega,\omega^2\},
\qquad
        \omega^2+\omega+1=0.
\]
For each primitive cube root \(\omega\in F\), the corresponding
condition is
\[
        ar^3+\bigl(b+\omega(b+d)\bigr)r+a^2=0,
\]
another nonzero cubic.  Thus \(L=0\) excludes at most six additional
values.  The field \(F\) contains primitive cube roots of unity
exactly when \(m\) is even, since \(3\mid 2^m-1\) exactly in that
case.  Thus for odd \(m\) this last exceptional set is empty.  In all cases, at most
twelve values of \(r\) are excluded, so a permissible \(r\) exists
for \(m\geq16\).  The first part of
Proposition~\ref{prop:fl-equal-a-genus} then gives genus \(127\).

If the projection has rank zero, write
\[
        \bfa^\T=(0,0,b)^\T,\qquad
        \bfb^\T=(0,0,d)^\T,\qquad b+d\neq0.
\]
Because \(\bfa^\T\) and \(\bfb^\T\) are nonzero, also \(b,d\neq0\); hence
\(K(K+e_r)\neq0\) for every nonzero \(r\).  Now
\[
        \frac{K}{e_r}=\frac{b}{b+d}
\]
is independent of \(r\).  Unless this ratio is a primitive cube root
of unity, \(L\neq0\) and the first part of
Proposition~\ref{prop:fl-equal-a-genus} applies.  In the exceptional
case, the second part applies after excluding \(e_r=1\), which is
equivalent to
\[
        r^5=b+d.
\]
This equation has at most five roots.  A permissible \(r\) again
exists, and the resulting genus is \(109\).  Thus the genus is at
most \(127\) in every rank-zero or rank-one case.
\end{proof}

\subsection{Rational points and the Hasse--Weil thresholds}

A rational point on a simultaneous splitting cover is useful only if
its image \(s\) is finite, avoids the branch locus, and lies outside
\(\{0,1\}\).  We first use a uniform subtraction for these invalid
fibers.

\begin{proposition}[Uniform thresholds]
\label{prop:fl-uniform-thresholds}
Let \(W\) be a \(555\) plane with zero first coordinate.

\begin{enumerate}
\item If \(\dim_{\F_2}\pi_3(W)=2\), then, for every \(m\geq18\),
      there are a basis \(\bfa^\T,\bfb^\T\) of \(W\) and
      \(r\in F^*\) such that \(\mathcal X_{\bfa,\bfb,r}\) has a good
      \(F\)-rational point.
\item If \(\dim_{\F_2}\pi_3(W)\leq1\), then, for every
      \(m\geq16\), there are a basis \(\bfa^\T,\bfb^\T\) of \(W\)
      and \(r\in F^*\) such that \(\mathcal X_{\bfa,\bfb,r}\) has a
      good \(F\)-rational point.
\end{enumerate}
\end{proposition}

\begin{proof}
In projection rank two, choose a basis with
\[
        a\neq0,\qquad c\neq0,\qquad a\neq c.
\]
Choose \(r\in F^*\) after excluding the at most six values from
Lemma~\ref{lem:fl-kappa-zero}.  Proposition~\ref{prop:fl-distinct-a-genus}
gives a degree-\(36\), genus-\(181\) cover.  At most \(36\) rational
points lie above infinity.  The four branch fibers contain at most
\[
        4\left(\frac{36}{2}\right)=72
\]
geometric points, and the two fibers \(s=0,1\) contain at most
\(2\cdot36=72\) points.  Thus at most \(180\) rational points are
discarded.  Hasse--Weil supplies a good point as soon as
\[
        q+1-2\cdot181\sqrt q>180.
\]
For powers of two this holds at \(q=2^{18}\), and the left-hand side
is increasing throughout the range \(q\geq2^{18}\).

If the projection rank is at most one, choose \(r\) as in
Lemma~\ref{lem:fl-low-rank-locator}.  The cover has degree \(36\) and
genus at most \(127\).  There are at most \(36\) points above
infinity, at most
\[
        2\left(\frac{36}{4}\right)=18
\]
above the two common branch fibers, and at most \(72\) above
\(s=0,1\).  The bad-point bound is therefore \(126\).  At
\(q=2^{16}\),
\[
        q+1-2\cdot127\sqrt q
        =
        65537-254\cdot256
        =
        513>126.
\]
The same inequality persists for every larger \(q\).
\end{proof}

The only case not covered by
Proposition~\ref{prop:fl-uniform-thresholds} in the desired algebraic
range is projection rank two at \(m=17\).  In odd extension degree a
trace choice eliminates the bad rational fibers instead of subtracting
them.

\begin{proposition}[Boundary case \(m=17\)]
\label{prop:fl-boundary-m17}
Let \(m=17\), and let
\[
        W=\{\bzero^\T,\bfa^\T,\bfb^\T,(\bfa+\bfb)^\T\}
\]
be a \(555\) plane with zero first coordinate.  If
\[
        \dim_{\F_2}\pi_3(W)=2,
\]
then there are a basis \(\bfa^\T,\bfb^\T\) of \(W\) and
\(r\in F^*\) such that \(\mathcal X_{\bfa,\bfb,r}\) has a good
\(F\)-rational point.
\end{proposition}

\begin{proof}
Choose a basis
\[
        \bfa^\T=(0,a,b)^\T,
        \qquad
        \bfb^\T=(0,c,d)^\T
\]
with \(a,c\neq0\) and \(a\neq c\).  Since \(m\) is
odd, the cube map permutes \(F^*\).  Put
\[
        \tau\eqdef r^{-3}.
\]
The two trace forms
\[
        \tau\longmapsto\tr(a\tau),
\qquad
        \tau\longmapsto\tr(c\tau)
\]
are linearly independent over \(\F_2\): otherwise the
nondegeneracy of the trace pairing would give \(a=c\).  Hence the
affine system
\[
        \tr(a\tau)=1,\qquad
        \tr(c\tau)=1
\]
has exactly \(q/4\) solutions, none of which is zero.  The bijection
\(r\mapsto r^{-3}\) therefore gives \(q/4\) candidate locators.
After the at most six exclusions in
Lemma~\ref{lem:fl-kappa-zero}, at least
\[
        \frac q4-6=32762
\]
choices remain.

Fix one such \(r\).  Then
\[
        \tr(a_r)
        =
        \tr(c_r)
        =
        1.
\]
The roots of \(s^2+s+a_r\) and \(s^2+s+c_r\) are nonrational, so none
of the four branch fibers lies over an \(F\)-rational base point.
At infinity, let \(\eta_{\bfa}\) and \(\eta_{\bfb}\) denote the residue
variables in the quadratic-resolvent extensions associated with
\(\bfa_r^\T\) and \(\bfb_r^\T\), respectively.
The infinity-residue assertion in
Proposition~\ref{prop:fl-one-cover} gives
\[
        \eta_{\bfa}^2+\eta_{\bfa}=a_r,
\qquad
        \eta_{\bfb}^2+\eta_{\bfb}=c_r.
\]
They have no solution in \(F\), so the full splitting cover has no
\(F\)-rational point above infinity.

There are also no rational points above \(s=0\) or \(s=1\).  At
\(s=0\), the prescribed pair is \(\{1,1\}\) and cancels in the binary
column sum.  A split completing cubic would therefore represent a
deep syndrome by at most three nonzero columns.  At \(s=1\), the pair
is \(\{1,0\}\); since \(h(0)=\bzero^\T\), a split completing cubic would
represent a deep syndrome by at most four nonzero columns.  Both
possibilities contradict depth five.

Thus every \(F\)-rational point of the smooth projective cover is
good.  Proposition~\ref{prop:fl-distinct-a-genus} gives genus \(181\),
and Hasse--Weil yields
\[
\begin{aligned}
 \#\mathcal X_{\bfa,\bfb,r}(F)
 &\geq
 2^{17}+1-2\cdot181\sqrt{2^{17}}\\
 &=131073-92672\sqrt2>0.
\end{aligned}
\]
The last inequality is exact, because
\[
        131073^2-2\cdot92672^2=3932161>0.
\]
Hence a good rational point exists.
\end{proof}

\subsection{From a good point to covering radius eight}

\begin{lemma}[Coding validity of a good point]
\label{lem:fl-good-point-support}
Let \(\bfa^\T,\bfb^\T,(\bfa+\bfb)^\T\) be deep.  A good rational point of
\(\mathcal X_{\bfa,\bfb,r}\) gives weight-five leaders of
\(\bfa^\T\) and \(\bfb^\T\) whose
intersection is exactly the prescribed pair.  Their support union
therefore has size eight.
\end{lemma}

\begin{proof}
At a good value of \(s\), the prescribed locators are distinct and
nonzero, both completing cubics split, and the point lies away from
the branch locus.  The completing roots are therefore distinct.
After translating the roots as in
Lemma~\ref{lem:fl-completing-cubic} and scaling back by \(r\), we
obtain two representations with five locators each.

A zero completing locator, or a collision between a completing
locator and the prescribed pair, would reduce one of these
representations to at most four nonzero columns, contrary to the
deepness of \(\bfa^\T\) or \(\bfb^\T\).  Thus both representations are genuine
weight-five leaders.  If their completing triples shared an additional
locator, then the symmetric difference of the two leaders would
represent \((\bfa+\bfb)^\T\) with at most four columns, again contradicting
depth.  The leaders consequently intersect in exactly two locators,
and their union has cardinality
\[
        5+5-2=8.
\]
\end{proof}

\begin{theorem}[Exact radius for \(m\geq17\)]
\label{thm:fl-stable-range}
For every \(m\geq17\),
\[
        R_2\!\left(\bch(3,m)\right)=8.
\]
\end{theorem}

\begin{proof}
Let \(\bfa^\T,\bfb^\T\in F^3\) be arbitrary, and put
\[
        W\eqdef
        \operatorname{span}_{\F_2}\{\bfa^\T,\bfb^\T\}.
\]
If \(\dim_{\F_2}W=0\), the claim is trivial.  If
\(\dim_{\F_2}W=1\), then one leader of its unique nonzero syndrome
covers both \(\bfa^\T\) and \(\bfb^\T\), and
\[
        d_2(\bfa^\T,\bfb^\T)\leq5.
\]
It remains to consider \(\dim_{\F_2}W=2\).  Suppose first that it is not a
\(555\) plane with zero first coordinate.  Then
the contrapositive of Lemma~\ref{lem:ca-zero-first-555} gives
\(d_2(\bfa^\T,\bfb^\T)\leq8\).

Now suppose that \(W\) is a \(555\) plane with zero first coordinate.  If
\(\dim_{\F_2}\pi_3(W)\leq1\), the second part of
Proposition~\ref{prop:fl-uniform-thresholds} supplies a good point for
every \(m\geq17\).  If \(\dim_{\F_2}\pi_3(W)=2\), the first part supplies a
good point for \(m\geq18\), while
Proposition~\ref{prop:fl-boundary-m17} supplies one for \(m=17\).
Lemma~\ref{lem:fl-good-point-support} then gives two leaders of a basis
of \(W\) whose support union has size eight.  Their symmetric
difference represents the third nonzero element of \(W\) and is still
contained in the same union.  Thus this eight-column union also
represents the originally chosen pair \(\bfa^\T,\bfb^\T\).  Hence
\[
        R_2(C_m)\leq8
        \qquad(m\geq17).
\]
The lower bound \eqref{eq:fl-published-lower-bound} now proves
equality.
\end{proof}

\section{The certified finite range}
\label{sec:finite-range-summary}

The point-counting argument proves the family theorem from \(m=17\)
onward.  The remaining dimensions are settled by the following exact
finite checks.  For \(m=5,6,7\), a direct whole-space scan checks every
relevant \(455\) and \(555\) syndrome plane, and an independent
XOR-convolution computation verifies the plane counts;
Lemma~\ref{lem:only-455-555} proves that these are the only possible
obstructions for every \(m\geq5\).  At \(m=8\),
Lemma~\ref{lem:ca-zero-first-555} reduces the problem to the subspace with zero
first coordinate, which is scanned directly.  For \(9\leq m\leq16\), locator scaling
reduces the \(555\) planes with zero first coordinate to finitely many
pairs of scaling orbits.  A completion signature records which normalized locator pairs
can be extended to a leader; intersecting two signatures tests whether
the leaders can share a locator pair.  The runs for
\(9\leq m\leq12\) check all pairs of scaling orbits directly, while the
runs for \(13\leq m\leq16\) additionally identify pairs related by the
Frobenius automorphism.

Appendix~\ref{sec:ca-finite-certificates} defines the signatures, proves
the orbit and transport rules, states the exact verifier criterion, and
records the exhaustive terminal counts.  In particular, the verifier
checks a complete deterministic orbit list, uses exact finite-field
arithmetic, and accepts only after every required orbit pair is covered.

\begin{theorem}[Certified finite range]
\label{thm:ca-certified-finite-range}
\[
        R_2\!\left(\bch(3,m)\right)=8
        \qquad(5\leq m\leq16).
\]
\end{theorem}

The proof depends on the exact verifier criterion and its orbit objects.
To keep every computational implication checkable where it is used, we
give the proof after Theorem~\ref{thm:ca-verifier-correctness} and
Table~\ref{tab:ca-finite-certificates} in
Appendix~\ref{sec:ca-finite-certificates}.

\section{Proof of the main theorem and outlook}

\begin{proof}[Proof of Theorem~\ref{thm:main}]
Theorem~\ref{thm:ca-certified-finite-range} gives the upper bound eight
for
\[
        5\leq m\leq16,
\]
and Theorem~\ref{thm:fl-stable-range} gives the same upper bound for
\[
        m\geq17.
\]
The lower bound eight for every \(m\geq5\) is
\cite[Theorem~3]{YohananovSchwartz2025}.  Hence equality holds
throughout the stated range.
\end{proof}

The structural reduction leaves a single possible obstruction: a
\(555\) plane with zero first coordinate.  We resolve it by converting the search
for two shared columns into
simultaneous splitting of two cubics.  Algebraic point counting
handles the infinite range, including the boundary dimension
\(m=17\), while exact certificates close the finite transition range.

Higher generalized covering radii require coordinating more than two
leaders; primitive BCH codes correcting more errors provide another
natural extension.  Subsequent work of Belinsky and Zabokritskiy
(Yohananov)~\cite{BelinskyZabokritskiy2026Higher} gives bounds for these
higher orders and code families.  Determining exact values in the
remaining parameter ranges is still open.  A uniform replacement for
the certificates used here for the finite transition range would also be
desirable, although the exact covering problem treated here is complete.

\appendix
\numberwithin{equation}{section}
\numberwithin{table}{section}
\renewcommand{\thetheorem}{\Alph{section}.\arabic{theorem}}

\section{Ramification details for the one-syndrome cover}
\label{sec:fl-ramification-details}

This appendix supplies the local function-field calculations used in
Proposition~\ref{prop:fl-one-cover}.  The first lemma determines the
ramified quadratic layer, and the second shows that the cyclic cubic
layer adds no further ramification.

Throughout this appendix,
\[
\begin{aligned}
 \bfs_r^\T&=(0,\alpha,\beta)^\T,&
 D&=\alpha+s^2+s,&
 N&=\beta+s^4+s,\\
 \kappa&=\beta+\alpha^2+\alpha\neq0,&
 \lambda&=\frac ND,&
 \mu&=D,
\end{aligned}
\]
as in Section~\ref{sec:fl-stable-range}.

\begin{lemma}[Quadratic resolvent and ramification]
\label{lem:fl-quadratic-resolvent}
The quadratic resolvent of \(y^3+\lambda y+\mu\) has
Artin--Schreier class
\[
        \mathcal H_\kappa\eqdef1+\frac{\lambda^3}{\mu^2}
        =1+\frac{N^3}{D^5}.
\]
It is ramified exactly at the two geometric zeros of \(D\), both of
which are simple.  At each of them its reduced pole order is five, so
its local different exponent and the Artin conductor exponent of its
unique nontrivial character are both six.  It is unramified at infinity.
\end{lemma}

\begin{proof}
For a depressed cubic in characteristic two, the standard alternating
quadratic resolvent is
\[
        z^2+z=1+\frac{\lambda^3}{\mu^2},
\]
see~\cite[Lemma~2.2]{MarquesWard2017Classification}; for the general
ramification framework for separable cubic function fields, see also
\cite{MarquesWard2019Ramification}.  Substitution of
\(\lambda=N/D\) and \(\mu=D\) gives the asserted class
\(\mathcal H_\kappa=1+N^3/D^5\).

Choose \(\theta\in\overline F\) with \(\theta^2=\kappa\), and write
\(\wp(z)\eqdef z^2+z\).  Direct reduction gives
\begin{equation}
\label{eq:fl-exact-AS-reduction}
\begin{aligned}
 \mathcal H_\kappa
 &+\wp\!\left(
       s+\frac1D+\frac{\kappa}{D^2}+\frac{\theta}{D}
     \right)\\
 &=
 \alpha
 +\frac{\kappa+\theta}{D}
 +\frac{\kappa^2+\kappa}{D^3}
 +\frac{\kappa^3}{D^5}.
\end{aligned}
\end{equation}
Since \(D'=1\), the polynomial \(D\) has two simple geometric zeros.
At either zero, the right-hand side of
\eqref{eq:fl-exact-AS-reduction} is reduced and has exact pole order
five, because \(\kappa\neq0\).  The Artin--Schreier different formula
\cite[Proposition~3.7.8]{Stichtenoth2009} therefore gives local
different exponent six.  Since the quadratic extension has a unique
nontrivial character, the conductor--discriminant formula gives Artin
conductor exponent six as well.

No other finite pole occurs.  At infinity put \(t\eqdef s^{-1}\).  Direct
expansion yields
\[
        \frac{N^3}{D^5}
        =
        t^{-2}+t^{-1}+(\alpha+1)+O(t).
\]
Thus
\[
        \mathcal H_\kappa+\wp(t^{-1})
        =
        \alpha+O(t),
\]
so infinity is unramified, with residue equation
\begin{equation}
\label{eq:fl-infinity-residue}
        \eta^2+\eta=\alpha.
\end{equation}
\end{proof}

Cancellation in the quadratic resolvent does not by itself rule out
tame cubic inertia.  We record the required local Cardano--Kummer
calculation explicitly.

\begin{lemma}[Geometric local splitting of the cyclic cubic layer]
\label{lem:fl-cardano-kummer}
Over the geometric quadratic resolvent field, the cyclic cubic layer
of the splitting field is unramified at every place.  More precisely,
after tensoring with the completion at any place of that field, it
splits above every zero of \(N\), above infinity, and above each zero
of \(D\).
\end{lemma}

\begin{proof}
The general Cardano--Kummer construction and its relation to the
alternating quadratic resolvent are described in
\cite[Theorem~1.6, Lemma~1.9(2), and Corollary~1.12]
{KaremakerMarquesSijsling2021}.  In the present notation, let \(\psi\)
be a root of
\begin{equation}
\label{eq:fl-cardano-quadratic}
        x^2+\mu x+\lambda^3.
\end{equation}
Then \(\psi_0\eqdef\psi/\mu\) satisfies
\[
        \psi_0^2+\psi_0=\frac{\lambda^3}{\mu^2}.
\]
After choosing \(\rho\in\overline F\) with \(\rho^2+\rho=1\), the
element \(\psi_0+\rho\) generates the alternating quadratic
resolvent.  Thus the Cardano quadratic and the alternating resolvent
coincide over the geometric base; over \(F(s)\) they may differ by the
constant Artin--Schreier class \(1\).

If \(\upsilon^3=\psi\) and \(\varphi\eqdef\lambda/\upsilon\), then
\[
        \varphi^3=\psi+\mu,
\]
and \(\upsilon+\varphi\) is a root of \(y^3+\lambda y+\mu\).
Since the geometric constant field contains the cube roots of unity,
adjoining \(\upsilon\) gives the full cyclic cubic layer.

In a completion of the geometric quadratic resolvent field, \(\psi\)
is a cube whenever \(3\mid v(\psi)\): its unit residue is a cube in the
algebraically closed residue field, and Hensel's lemma lifts the
residue root.  Compare
\cite[Theorem~3.3.7 and Proposition~3.7.3]{Stichtenoth2009} for place
splitting and Kummer ramification.

The Newton polygons of \eqref{eq:fl-cardano-quadratic} give
\[
\begin{array}{c|ccc|c}
 \text{place}
 &v(\lambda)&v(\mu)&v(\lambda^3)
 &\text{valuations of the two choices of }\psi\\ \hline
 N=0      & 1  & 0  & 3  & 0,\ 3\\
 \infty   &-2  &-2  &-6  &-3,\ -3\\
 D=0      &-1  & 1  &-3  &-\frac32,\ -\frac32
\end{array}
\]
Here the zeros of \(N\) are simple and \(D\) is a unit there; at a
zero of \(D\), one has \(N=\kappa\neq0\).  By
Lemma~\ref{lem:fl-quadratic-resolvent}, the quadratic resolvent has
ramification index two at \(D=0\), so the last two valuations become
\(-3\) in its normalized valuation.  Thus, at every place of the
geometric quadratic resolvent field, the valuation of \(\psi\) is
divisible by three.  Away from the listed places, each nonzero Cardano
element is a unit.  The local Kummer criterion therefore shows that
the cyclic cubic layer splits at every completion.
\end{proof}

\begin{proof}[Proof of Proposition~\ref{prop:fl-one-cover}]
At either zero of \(D\), the Newton polygon of
\(Dy^3+Ny+D^2\) gives root valuations
\[
        2,\qquad-\frac12,\qquad-\frac12.
\]
Since \(\gcd(D,N)=1\), a rational root \(a_0/b_0\), written in lowest
terms, must satisfy
\[
        a_0\mid D^2,\qquad b_0\mid D.
\]
Its valuations at the two zeros of \(D\) are integral and hence both
equal two.  Therefore \(b_0\) is constant and
\(a_0/b_0=c_0D^2\).  Substitution gives
\[
        c_0^3D^5+c_0N+1=0,
\]
which is impossible by degrees.  The cubic is therefore geometrically
irreducible.  It is separable because a common root with its derivative
\(y^2+\lambda\) would force \(\mu=0\).

For an irreducible separable cubic, the quadratic resolvent
distinguishes the groups \(A_3\) and \(S_3\); see
\cite[Theorem~2.3]{MarquesWard2017Classification}.  The ramified,
nontrivial resolvent in Lemma~\ref{lem:fl-quadratic-resolvent} therefore
gives geometric Galois group \(S_3\).  The geometric splitting field
already has the maximal possible degree six, so the arithmetic group
is also \(S_3\), and the cover is regular.

Lemmas~\ref{lem:fl-quadratic-resolvent} and
\ref{lem:fl-cardano-kummer} show that the only ramification occurs at
the two zeros of \(D\), with inertia \(C_2\) and local different
exponent six.  Riemann--Hurwitz
\cite[Theorem~3.4.13]{Stichtenoth2009} gives
\[
        2g(\mathcal X_{\bfs_r})-2
        =-12+2\left(\frac{6}{2}\right)6=24,
\]
and hence \(g(\mathcal X_{\bfs_r})=13\).
\end{proof}

\section{Exact verification of the finite range}
\label{sec:ca-finite-certificates}

The algebraic point-counting argument applies for \(m\geq17\).  We
settle \(5\leq m\leq16\) by exact finite certificates.  This appendix
gives the mathematical reduction and the criterion verified by those
certificates; implementation details, complete outputs, and checksums
are archived in~\cite{Yohananov2026VerificationArtifact}.

\subsection{Completion signatures and scaling}

Identify a syndrome with zero first coordinate, \((0,a,b)^\T\in F^3\), with
\((a,b)^\T\in F^2\).  We first describe the set \(\mathcal T\) of
such syndromes represented by three distinct nonzero locators:
\[
\begin{aligned}
 \mathcal T\eqdef\{\,(e,pe)^\T:\;&e\in F^*,\ p\in F,\\
 &y^3+py+e
 \text{ splits into three distinct linear factors over }F\,\}.
\end{aligned}
\]
Indeed, if \(\gamma_1,\gamma_2,\gamma_3\) are the roots of
\(y^3+py+e\), Newton's identities in characteristic two give
\[
 \sum_i\gamma_i=0,
 \qquad
 \sum_i\gamma_i^3=e,
 \qquad
 \sum_i\gamma_i^5=pe.
\]
Conversely, a zero-sum triple with third and fifth power sums \(e\) and
\(pe\) has this locator polynomial.  Such a representation
is unique: two distinct triples with the same syndrome would have
symmetric difference equal to a nonzero codeword of weight at most six,
contrary to the BCH bound.  By the elementary fact preceding
Lemma~\ref{lem:ca-zero-first-555}, the set of deep syndromes with zero first coordinate is
\[
        \mathcal D\eqdef F^2\setminus
        \bigl(\mathcal T\cup\{\bzero^\T\}\bigr).
\]

To describe the pair of locators that two leaders will share, write
\[
 s=\gamma_1+\gamma_2\neq0,
 \qquad
 \delta=\frac{\gamma_1\gamma_2}{s^2}.
\]
Every pair of distinct nonzero locators \(\{\gamma_1,\gamma_2\}\) is
encoded uniquely in this way, with \(\delta\neq0\) and
\(\tr(\delta)=0\); conversely, it is the root set of
\(x^2+sx+\delta s^2\).  Put
\[
 \Delta\eqdef\{\delta\in F^*: \tr(\delta)=0\},
 \qquad
 R_\delta(s)\eqdef
 \bigl(\delta s^3,(\delta+\delta^2)s^5\bigr)^\T
 \quad(s\in F^*).
\]
The two roots of \(x^2+sx+\delta s^2\), together with \(s\), are a
zero-sum locator triple.  Their third and fifth power sums are
\(\delta s^3\) and \((\delta+\delta^2)s^5\), respectively, so their
syndrome, whose first coordinate is zero, is exactly \(R_\delta(s)\).  To record which of
these pairs can be completed to a weight-five leader of
\(\bfa^\T\in\mathcal D\), define its completion signature by
\begin{equation}
\label{eq:ca-signature}
 \Sigma_{\bfa}(\delta)
 \eqdef
 \{s\in F^*:R_\delta(s)+\bfa^\T\in\mathcal T\}.
\end{equation}
Locator scaling acts by
\begin{equation}
\label{eq:ca-scaling-action}
 \lambda\cdot(a,b)^\T
 \eqdef(\lambda^3a,\lambda^5b)^\T,
 \qquad \lambda\in F^*,
\end{equation}
and preserves \(\mathcal T\) and \(\mathcal D\).  For arbitrary
\(\mathcal U,\mathcal V\subseteq F^*\), write
\(\mathcal U\mathcal V^{-1}\eqdef
\{uv^{-1}:u\in\mathcal U,\ v\in\mathcal V\}\).

\begin{lemma}[Signature quotient criterion]
\label{lem:ca-signature-quotient}
Let \(\bfa^\T,\bfb^\T\in\mathcal D\).
If, for some \(\delta\in\Delta\),
\[
        \Sigma_{\bfa}(\delta)\cap
        \Sigma_{\bfb}(\delta)\neq\varnothing,
\]
then \(\bfa^\T\) and \(\bfb^\T\) have weight-five leaders sharing at least two
locators, and hence \(d_2(\bfa^\T,\bfb^\T)\leq8\).  More generally, the
same conclusion holds for \(\bfa^\T\) and
\(\mu\cdot\bfb^\T\) whenever
\[
 \mu\in Q_{\bfa,\bfb}(\delta)
 \eqdef
 \Sigma_{\bfa}(\delta)\Sigma_{\bfb}(\delta)^{-1}.
\]
\end{lemma}

\begin{proof}
If \(s\in\Sigma_{\bfa}(\delta)\), translating the zero-sum triple
represented by \(R_\delta(s)+\bfa^\T\) by \(s\) gives the three
locators completing this pair to a representative of \(\bfa^\T\).
Deepness excludes zero or a collision among the resulting five
locators, since cancellation would give a representative of weight at
most four.  Hence a common signature entry gives two weight-five
supports with union of size at most eight.  Finally,
\[
 \Sigma_{\mu\cdot\bfb}(\delta)=
 \mu\Sigma_{\bfb}(\delta),
\]
so the two signatures intersect exactly when
\(\mu=st^{-1}\) for some
\(s\in\Sigma_{\bfa}(\delta)\) and
\(t\in\Sigma_{\bfb}(\delta)\).
\end{proof}

The scaling orbits also admit a short exact parametrization.

\begin{lemma}[Complete representatives of scaling orbits]
\label{lem:ca-scaling-orbits}
On \(ab\neq0\),
\[
        I(a,b)\eqdef\frac{a^5}{b^3}
\]
is a complete invariant for~\eqref{eq:ca-scaling-action}.  Every such
orbit is free and has the unique diagonal representative
\((u,u)\), where \(u^2=I(a,b)\).  On \(b=0\) and \(a=0\), respectively,
the orbits are the cosets of \((F^*)^3\) and \((F^*)^5\); their numbers,
and their stabilizer orders, are \(\gcd(3,q-1)\) and
\(\gcd(5,q-1)\).
\end{lemma}

\begin{proof}
The identity
\(I(\lambda^3a,\lambda^5b)=I(a,b)\) is immediate.  Since squaring is
an automorphism of \(F\), the unique choice
\(\lambda=\sqrt{a/b}\) sends \((a,b)\) to \((u,u)\) with
\(u^2=a^5/b^3\); this also proves completeness and uniqueness.
A generic stabilizer satisfies \(\lambda^3=\lambda^5=1\), hence is
trivial, and the two axis assertions are the corresponding power-coset
decompositions of the cyclic group \(F^*\).
\end{proof}

\subsection{The exact cubic test}

The following cubic splitting criteria are classical~\cite{BerlekampRumseySolomon1967}.
Depressed-cubic and
Artin--Schreier trace tests in BCH locator calculations also appear in
van der Horst and Berger~\cite[Appendix~D]{VanDerHorstBerger1976};
compare~\cite{YohananovSchwartz2025,OzbudakOzturk2026Second}.  We retain
the short normalization because it fixes the precise convention used
by the certificates.

\begin{lemma}[Splitting criterion for normalized cubics]
\label{lem:ca-normalized-cubic}
Let \(f(y)\eqdef y^3+py+e\), where \(e\neq0\).  If \(p\neq0\), let
\(d^2=p\) and put \(c\eqdef e/d^3=e/(dp)\).  Then \(f\) splits into
three distinct linear factors over \(F\) if and only if
\(x^3+x+c\) does.  For \(c\neq0\), the latter condition is equivalent
to the existence of \(\zeta\in F\setminus\{0,1\}\) such that
\[
        c=\zeta^3+\zeta,
        \qquad
        \tr(1+\zeta^{-2})=0.
\]
If \(p=0\), the polynomial \(y^3+e\) never splits completely when \(m\)
is odd; when \(m\) is even, it splits distinctly if and only if \(e\)
is a cube.  The case \(e=0\) is rejected because the cubic has a
repeated root.
\end{lemma}

\begin{proof}
For \(p\neq0\), substitution \(y=dx\) gives
\(f(dx)=d^3(x^3+x+c)\).  If \(\zeta\) is a root, then
\[
 x^3+x+c=(x+\zeta)
 (x^2+\zeta x+\zeta^2+1),
\]
and the quadratic factor splits precisely when the Artin--Schreier
equation \(t^2+t+(1+\zeta^{-2})=0\) is solvable, which is the displayed
trace condition.  The \(p=0\) assertion follows from the presence of
all three cube roots of unity in \(F\), equivalently \(m\) even.
\end{proof}

For direct signature construction, write
\[
 \bfa^\T=(a,b)^\T,
 \qquad
 \alpha\eqdef a/s^3,
 \qquad
 \eta\eqdef b/s^5,
\]
and define
\begin{equation}
\label{eq:ca-completion-cubic}
 e\eqdef\delta+\alpha,
 \qquad \text{when }e\neq0,\qquad
 w\eqdef\frac{\eta+\alpha+\alpha^2}{e},
 \qquad
 p\eqdef1+w+e.
\end{equation}
A direct substitution in~\eqref{eq:ca-signature} gives
\[
 s\in\Sigma_{\bfa}(\delta)
 \quad\Longleftrightarrow\quad
 e\neq0\ \text{ and }\ y^3+py+e
 \text{ splits distinctly over }F.
\]
Thus Lemma~\ref{lem:ca-normalized-cubic} is an exact signature test.

\subsection{Frobenius transport and verifier correctness}

Frobenius symmetry allows us to check one pair of representatives of scaling
orbits and transfer its coverage condition to every pair in
the same diagonal Frobenius orbit.  We describe this transfer before
stating the complete verification criterion.

For \(\bfa^\T=(a,b)^\T\), write
\((\bfa^{[2]})^\T\eqdef(a^2,b^2)^\T\), and for
\(\mathcal U\subseteq F\), put
\(\mathcal U^2\eqdef\{u^2:u\in\mathcal U\}\).  Frobenius invariance of
\(\mathcal T\), together with
\(R_\delta(s)^{[2]}=R_{\delta^2}(s^2)\), gives
\begin{equation}
\label{eq:ca-raw-frobenius}
 \Sigma_{\bfa^{[2]}}(\delta^2)=
 \Sigma_{\bfa}(\delta)^2.
\end{equation}
By Lemma~\ref{lem:ca-scaling-orbits}, fix a complete family
\((\bfa_i^\T)_{i\in I}\) of representatives of the scaling orbits of deep
syndromes with zero first coordinate, including the axis orbits, and record the stabilizer of
each representative.  Frobenius induces a permutation \(\sigma\) of the
finite index set \(I\); choose \(\lambda_i\in F^*\) such that
\begin{equation}
\label{eq:ca-representative-correction}
        (\bfa_i^{[2]})^\T
        =\lambda_i\cdot\bfa_{\sigma(i)}^\T,
\end{equation}
and abbreviate
\(Q_{i,j}(\delta)\eqdef Q_{\bfa_i,\bfa_j}(\delta)\).

\begin{lemma}[Frobenius transport of signatures]
\label{lem:ca-frobenius-transport}
For every \(i,j\) and \(\delta\in\Delta\),
\begin{align}
\label{eq:ca-signature-transport}
 \Sigma_{\bfa_{\sigma(i)}}(\delta^2)
 &=\lambda_i^{-1}\Sigma_{\bfa_i}(\delta)^2,\\
\label{eq:ca-quotient-transport}
 Q_{\sigma(i),\sigma(j)}(\delta^2)
 &=\frac{\lambda_j}{\lambda_i}Q_{i,j}(\delta)^2.
\end{align}
These identities are independent of the choice of \(\lambda_i\) when
the image orbit has a nontrivial stabilizer.
\end{lemma}

\begin{proof}
Scaling covariance and~\eqref{eq:ca-raw-frobenius} give
\[
 \Sigma_{\bfa_i}(\delta)^2
 =\Sigma_{\bfa_i^{[2]}}(\delta^2)
 =\lambda_i\Sigma_{\bfa_{\sigma(i)}}(\delta^2),
\]
which proves~\eqref{eq:ca-signature-transport}; taking quotient sets
gives~\eqref{eq:ca-quotient-transport}.  Two permissible correction
scalars differ by a stabilizer, under which the corresponding signature
is invariant.
\end{proof}

For each stored pair \((i,j)\), let
\[
 \mathcal Q_{i,j}\eqdef
 \bigcup_{\delta\in\Delta_{i,j}}Q_{i,j}(\delta),
\]
where \(\Delta_{i,j}\subseteq\Delta\) is the deterministic list of
trace-zero parameters scanned for that pair, and define the set of required
scalings not yet covered by
\begin{equation}
\label{eq:ca-hard-missing-set}
 \mathcal M_{i,j}\eqdef
 \left\{\mu\in F^*:
 \begin{array}{l}
 \bfa_i^\T,\ \mu\cdot\bfa_j^\T,\
 \bfa_i^\T+\mu\cdot\bfa_j^\T
 \text{ are distinct and deep},\\[-1mm]
 \mu\notin\mathcal Q_{i,j}
 \end{array}
 \right\}.
\end{equation}

\begin{theorem}[Verifier correctness]
\label{thm:ca-verifier-correctness}
Fix \(m\geq8\).  Suppose that an exact finite verification has the
following three properties.
\begin{enumerate}
\item The stored vectors \(\bfa_i^\T\) form a complete set of
representatives of the scaling orbits of deep syndromes with zero first coordinate, including all
axis orbits, with their stabilizers recorded.
\item Every unordered pair, with repetition, of scaling orbits is
represented either directly or by one representative of its diagonal
Frobenius orbit.  In the latter case the verified scalars \(\lambda_i\)
satisfy~\eqref{eq:ca-representative-correction}, and quotient sets are
transported by~\eqref{eq:ca-quotient-transport}, with inversion when an
ordered lift is reversed.
\item Every signature used in forming \(\mathcal Q_{i,j}\) is evaluated
exactly and exhaustively; parameters may be processed deterministically
and adaptively until \(\mathcal M_{i,j}=\varnothing\).  The
implementation uses either direct pair--triple enumeration based on
\eqref{eq:ca-signature}, or the equivalent cubic test given by
\eqref{eq:ca-completion-cubic} and
Lemma~\ref{lem:ca-normalized-cubic}.  For every checked pair
representative, the output satisfies
\[
        \mathcal M_{i,j}=\varnothing.
\]
\end{enumerate}
Then every \(555\) plane with zero first coordinate has two weight-five
representatives whose support union has size at most eight.  Consequently,
\[
        R_2(C_m)\leq8.
\]
\end{theorem}

\begin{proof}
Let
\(\{\bzero^\T,\bfa^\T,\bfb^\T,(\bfa+\bfb)^\T\}\)
be a \(555\) plane with zero first coordinate.  A global locator scaling sends
\(\bfa^\T\) to some \(\bfa_i^\T\), while the image of \(\bfb^\T\) has
the form \(\mu\cdot\bfa_j^\T\).  For a directly checked pair,
\(\mathcal M_{i,j}=\varnothing\) implies
\(\mu\in\mathcal Q_{i,j}\), so
Lemma~\ref{lem:ca-signature-quotient} gives the required common locator
pair.  For a Frobenius image, equations
\eqref{eq:ca-signature-transport}--\eqref{eq:ca-quotient-transport}
transport both the quotient and its complement bijectively (and pair
reversal only inverts \(\mu\)); hence an uncovered required scaling at the
image would pull back to one at its checked representative.  Thus every
\(555\) plane with zero first coordinate is covered, and
Lemma~\ref{lem:ca-zero-first-555} gives the conclusion for all syndrome
pairs.
\end{proof}

The certificates establish the condition
\(\mathcal M_{i,j}=\varnothing\) in Theorem~\ref{thm:ca-verifier-correctness}:
\(\mathcal Q_{i,j}\) contains every relative scaling for which the three
syndromes are distinct and deep.  It may omit scalings giving a shallow
third syndrome or a degenerate pair.

\subsection{Finite results}

For \(m=5,6,7\), whole-space computations enumerate all supports of
size at most four, use the known ordinary covering radius five
\cite{VanDerHorstBerger1976,AssmusMattson1976,Helleseth1978} to identify
the deep syndromes, and test every \(455\) and \(555\) plane.  For
\(m=8\), Lemma~\ref{lem:ca-zero-first-555} reduces the computation to a
direct scan of the \(555\) planes with zero first coordinate.  For
\(9\leq m\leq16\), the certificates satisfy
Theorem~\ref{thm:ca-verifier-correctness}.  The exact terminal data are
summarized in the single table below.

For the orbit-based rows, \(N_{\mathrm{orb}}\) denotes the number of scaling
orbits of deep syndromes with zero first coordinate.

\begin{table}[ht]
\centering
\small
\begin{tabular}{@{}clrrr@{}}
\toprule
\(m\)&system&\(N_{\mathrm{orb}}\)&checked (expanded)&terminal\\
\midrule
 5&whole-space planes&--&\(16{,}021{,}544/43{,}747{,}200\)&0\\
 6&whole-space planes&--&\(123{,}921{,}896/1{,}112{,}999{,}328\)&0\\
 7&whole-space planes&--&\(28{,}914{,}344/519{,}851{,}640\)&0\\
 8&planes in \(\{0\}\times F^2\)&--&\(416{,}704{,}680\)&0\\
 9&pairs of scaling orbits&428&\(91{,}806\)&0\\
10&pairs of scaling orbits&856&\(366{,}796\)&0\\
11&pairs of scaling orbits&1,708&\(1{,}459{,}486\)&0\\
12&pairs of scaling orbits&3,420&\(5{,}849{,}910\)&0\\
13&diagonal Frobenius orbits&6,828&\(1{,}793{,}406\ (23{,}314{,}206)\)&0\\
14&diagonal Frobenius orbits&13,656&\(6{,}661{,}282\ (93{,}249{,}996)\)&0\\
15&diagonal Frobenius orbits&27,308&\(24{,}858{,}504\ (372{,}877{,}086)\)&0\\
16&diagonal Frobenius orbits&54,620&\(93{,}233{,}211\ (1{,}491{,}699{,}510)\)&0\\
\bottomrule
\end{tabular}
\caption{Exact finite verification.  For \(m=5,6,7\), the two checked
counts are the numbers of \(555\) and \(455\) planes, and the terminal
entry counts uncovered planes.  For \(m=8\), it counts uncovered
\(555\) planes with zero first coordinate.  For \(m\geq9\), it is
\(\sum_{i,j}|\mathcal M_{i,j}|\).  For \(m=13,14,15,16\), the
parenthetical count is the number of unordered pairs of scaling orbits
after expansion of the checked diagonal Frobenius orbits.}
\label{tab:ca-finite-certificates}
\end{table}

\begin{proof}[Proof of Theorem~\ref{thm:ca-certified-finite-range}]
For \(m=5,6,7\), Lemma~\ref{lem:only-455-555} shows that only the
\(455\) and \(555\) patterns can obstruct the upper bound.  The exhaustive
whole-space scan finds no uncovered plane of either type, and an
independent XOR-convolution computation reproduces the two plane counts.
For \(m=8\), the direct scan in \(\{0\}\times F^2\), together with
Lemma~\ref{lem:ca-zero-first-555}, gives the upper bound eight.  For
\(9\leq m\leq16\), the archived orbit runs satisfy every hypothesis of
Theorem~\ref{thm:ca-verifier-correctness}.  For \(m=16\), the stand-alone
aggregator additionally verifies that the archived certificate
intervals form a gap-free, overlap-free partition of the
verifier-declared representative-index interval and recomputes the
recorded aggregate counts.  It does not independently regenerate the
orbit universe or re-evaluate the quotient predicate; those steps are
supported by the archived verifier source and production run.  Hence
the upper bound is eight throughout the stated range.  The matching
lower bound is \cite[Theorem~3]{YohananovSchwartz2025}.
\end{proof}

The complete source, exact outputs, manifests, certificates, and the
stand-alone \(m=16\) aggregator are preserved in the archival
record~\cite{Yohananov2026VerificationArtifact}, together with the field
models and implementation checks.  The proof uses the complete exact
production runs in this record.

\section*{Data availability}

The exact release supporting the finite-range part of this paper is
Version~1.0.0 in the Zenodo repository
\href{https://doi.org/10.5281/zenodo.21796632}
{\nolinkurl{10.5281/zenodo.21796632}}~\cite{Yohananov2026VerificationArtifact}.
SHA-256 checksums verify file integrity.  The roles of the finite-field
programs and the \(m=16\) certificate aggregator are specified in the
proof of Theorem~\ref{thm:ca-certified-finite-range}.

\bibliographystyle{elsarticle-num}
\bibliography{references}

\end{document}